\documentclass[11pt]{article}
\usepackage{xcolor}
\usepackage{amsmath,amssymb,amsthm}
\usepackage[nocompress]{cite}
\usepackage{geometry}
\usepackage{enumitem}
\usepackage{tikz}
\theoremstyle{plain} 
\newtheorem{theorem}{Theorem}
\newtheorem{lemma}[theorem]{Lemma} 
\newtheorem{corollary}[theorem]{Corollary} 

\theoremstyle{definition} 
\newtheorem{example}{Example} 
\usepackage{graphicx} 

\newcommand{\lca}{\ensuremath\mathsf{lca}}

\newcommand{\poly}{\ensuremath\mathrm{poly}}

\title{Improved Extended Regular Expression Matching}
\author{Philip Bille \\\texttt{phbi@dtu.dk} \and Inge Li G{\o}rtz \\\texttt{inge@dtu.dk} \and Rikke Schjeldrup Jessen \\\texttt{rscje@dtu.dk}}

\begin{document}
\maketitle
 \thispagestyle{empty}
\begin{abstract}
An extended regular expression $R$ specifies a set of strings formed by characters from an alphabet combined with concatenation, union, intersection, complement, and star operators. Given an extended regular expression $R$ and a string $Q$, the extended regular expression matching problem is to decide if $Q$ matches any of the strings specified by $R$. Extended regular expressions are a basic concept in formal language theory and a basic primitive for searching and processing data. 
Extended regular expression matching was introduced by Hopcroft and Ullman in the 1970s [\textit{Introduction to Automata Theory, Languages and Computation}, Addison-Wesley, 1979], who gave a simple dynamic programming solution using $O(n^3m)$ time and $O(n^2m)$ space, where $n$ is the length of $Q$ and $m$ is the length of $R$. Since then, several solutions have been proposed, but few significant asymptotic improvements have been obtained. The current state-of-the-art solution, by Yamamoto and Miyazaki~[COCOON, 2003], uses $O(\frac{n^3k + n^2m}{w} + n + m)$ time and $O(\frac{n^2k + nm}{w} + n + m)$ space, where $k$ is the number of intersection and complement operators in $R$ and $w$ is the number of bits in a machine word. This roughly replaces the $m$ factor with $k$ in the dominant terms of both the space and time bounds of the classical Hopcroft and Ullman algorithm. 

In this paper, we revisit the problem and present a new solution that significantly improves the previous time and space bounds. Our main result is a new algorithm that solves extended regular expression matching in 
    \[
    O\left(n^\omega k + \frac{n^2m}{\max(w/\log w, \log (n + m))} + m\right)
    \]
    time and $O(\frac{n^2 \log k}{w} + n + m) = O(n^2 +m)$ space, where $\omega \approx 2.3716$ is the exponent of matrix multiplication. Essentially, this replaces the dominant $n^3k$ term with $n^\omega k$ in the time bound, while simultaneously improving the $n^2k$ term in the space to $O(n^2)$. We also consider the interval operator (also known as the counting operator) and show how to easily and efficiently extend our algorithms to this case.  

Our results are based on a surprisingly simple combination of techniques and insights, including a compact representation to store and efficiently combine substring matches, a clustering technique for parse trees of extended regular expressions, and a new efficient combination of finite automaton simulation with our substring match representation to speed up the classic dynamic programming solution.

\end{abstract}

\newpage
\setcounter{page}{1}
\section{Introduction}
An extended regular expression $R$ specifies a set of strings formed by characters from an alphabet combined with concatenation ($\odot$), union ($\mid$), intersection ($\cap$), complement ($\neg$), and star ($^\ast$) operators. For instance, $\neg((a|b)^*)\odot b$ describes the set of strings that end in $b$ and contain at least one character that is not $a$ or $b$. Given an extended regular expression $R$ and a string $Q$, the extended regular expression matching problem is to decide if $Q$ matches any of the strings specified by $R$. The intersection and complement are the extended operators, and without these, the expression is simply called a \emph{regular expression}. 
Regular expressions are a fundamental concept in formal language theory introduced by Kleene in the 1950s~\cite{Kleene1956}, and extended and non-extended regular expression matching is a basic tool in computer science for searching and processing text. Standard tools such as \texttt{grep} and \texttt{sed} provide direct support for extended regular expression matching in files, and the scripting language \texttt{perl}~\cite{Wall1994} is a full programming language designed to support extended regular expression matching easily. Extended regular expression matching appears in many large-scale data processing applications such as internet traffic analysis~\cite{JMR2007, YCDLK2006, KDYCT2006, GWXC+2023, MPNT+2010, TSCV2004, BGP2025, AFSK+2015, BC2013}, data mining~\cite{GRS1999, AB2003, TBG2008, AGJ2007}, databases~\cite {LM2001, Murata2001, HP2003, MW1995, FLMTW2011}, computational biology~\cite{Myers1996, NR2003, AH2006, XWD2004, ER2002}, and human-computer interaction~\cite{KHDA2012, SCP2012, SCPF2013}.

The non-extended regular expression matching problem is a well-studied central problem in pattern matching~\cite{Thompson1968, Galil1985, Myers1992, BFC2008, Bille2006, BT2010, BT2010, BG2022, BG2024, BM2013, DGGS2022, BI2016, BGL2017, Schepper2020}. A classic textbook solution from 1968 by Thompson~\cite{Thompson1968} solves the problem in $O(nm)$ time and $O(m)$ space, where $n$ is the length of the string $Q$ and $m$ is the length of $R$. We can improve this time bound by polylogarithmic factors~\cite{Myers1992, BFC2008, Bille2006, BT2009, BT2010}, but by recent conditional lower bounds, we should not hope to achieve $O((nm)^{1-\epsilon})$ time for any $\epsilon>0$~\cite{BI2016, BGL2017, Schepper2020}. Hence, the complexity of non-extended regular expression matching in terms of $n$ and $m$ is essentially settled up to polylogarithmic factors.

The extended regular expression matching problem has also attracted significant research interest~\cite{HU1979, Hirst1989, KMb1995, KZ2002, Yamamoto2001, Yamamoto2000, ISY2003, YM2003, RV2003, Rocsu2005, Rocsu2007, Petersen2002, NT2026}, but in contrast to non-extended regular expression matching, surprisingly little is known about the complexity of the problem. The problem was first introduced in 1979 by Hopcroft and Ullman~\cite{HU1979} as a textbook exercise (see exercise 3.23 and the associated solution). They presented a simple dynamic programming algorithm that runs in $O(n^3m)$ time and uses $O(n^2m)$ space. However, with careful implementation and modern matrix multiplication algorithms, it is straightforward to improve the bound to $O(n^\omega m)$ time and $O(n^2m/w + n + m)$ space. Here, $\omega \approx 2.3713$ is the exponent of matrix multiplication~\cite{ADWXXZ2025} and $w$ is the length of a machine word. Apparently, these improvements have been overlooked by later papers. We review this algorithm in detail and show how to implement it efficiently in Section~\ref{sec:dynamicprogramming}.

Several improved algorithms have been proposed~\cite{Hirst1989, KMb1995, KZ2002, Yamamoto2001, Yamamoto2000, ISY2003, YM2003,RV2003, Rocsu2005, Rocsu2007}. Hirst in 1989~\cite{Hirst1989}, and later Kupferman and Zuhovitzky in 2002~\cite{KZ2002}, and Ilie, Shan, and Yu in 2003~\cite{ISY2003} all proposed solutions that run in $O(n^2m)$ time. Unfortunately, these results were all found to contain errors, and their claimed complexity bounds were incorrect. See Ro{\c{s}}u~\cite{Rocsu2007, Rocsu2005} for a detailed discussion of this and the conditional lower bound discussed below.  

Only a few improvements of the classic dynamic programming solution (without errors) are known. Let $n$ and $m$ be the lengths of the string and the extended regular expression, respectively. Yamamoto and Miyazaki~\cite{YM2003} gave an $O(\frac{n^3k + n^2m}{w} + n + m)$ time and $O(\frac{n^2k + nm}{w} + n + m)$ space algorithm, where $k$ is the number of intersection and complement operators in the regular expression. Essentially, if we ignore the $w$ factor, this result replaces the $m$ factor with $k$ in the dominant term of both the $O(n^3k)$ time and $O(n^2m)$ space bounds for most combinations of parameters. This result uses techniques different from the original dynamic programming algorithm, and it does not appear that the bound can be easily improved using matrix multiplication. Yamamoto~\cite{Yamamoto2001} also gave an $O(n^3\lambda + n^2m)$ time and $O(n^2 \lambda + nm)$ space algorithm, where $\lambda$ is the number of \emph{critical modules} of the regular expression. In general, a critical module is a subexpression of $R$ that contains a single complement or intersection operator. However, a module is only critical depending on its placement relative to other modules and star operators, and hence $\lambda \leq k$. Thus, compared to Yamamoto and Miyazaki~\cite{YM2003} (again ignoring the $w$ factor), the $k$ factor is replaced by $\lambda$. As before, it does not appear that this bound can be improved using matrix multiplication.  Finally, Ro{\c{s}}u~\cite{Rocsu2007} proposed an incomparable result that uses exponential time and space in $m$. 

On the hardness side, Petersen~\cite{Petersen2002} showed that regular expression matching extended with only the intersection operator is LOGCFL-complete. Very recently (and after the initial announcement of our result), Nogami and Terauchi~\cite{NT2026} showed that there is no $O(n^{\omega - \epsilon}\poly(m))$ time algorithm for any $\epsilon >0$ unless the $k$-clique hypothesis is false ($\poly(m)$ denotes a constant degree polynomial in $m$). Thus, under this hypothesis, we cannot improve the $\omega$ exponent on the dominant $n^\omega$ term by any constant. Note that the earlier incorrect $O(n^2m)$ time solutions break this lower bound. Nogami and Terauchi also note that a weaker conditional lower bound can be achieved via a recent reduction from $3k$-Clique to context-free language recognition by Abboud et al.~\cite{ABW2018}. See Table~\ref{tab:previousresults} for an overview of the bounds.

\begin{table}[t]
    \centering
    \renewcommand{\arraystretch}{1.2}
    \resizebox{\textwidth}{!}{%
    \begin{tabular}{|c|c|l|}
    \hline
    Time & Space & \\
    \hline
    $O(n^3m)$ & $O(n^2m)$ & HopCroft and Ullman~\cite{HU1979} \\
    $O(n^\omega m)$ & $O(n^2m)$ & HopCroft and Ullman~\cite{HU1979} w. matrix mult.\\
    $O(\frac{n^3k + n^2m}{w} + n + m)$ & $O(\frac{n^2k + nm}{w} + n + m)$ & Yamamoto and Miyazaki~\cite{YM2003}\\
    $O(n^3\lambda + n^2m)$ & $O(n^2 \lambda + nm)$ & Yamamoto~\cite{Yamamoto2001}\\
    $\Omega(n^{\omega - \epsilon}\poly(m))$ & & Nogami and Terauchi~\cite{NT2026} \\
    $O\left(n^\omega k + n^2m\right)$ & $O(n^2+ m)$ & \textbf{This paper}\\
    $O\left(n^\omega k + \frac{n^2m}{\max(w/\log w, \log (n+m))} + m\right)$& $O(n^2+ m)$ & \textbf{This paper}\\
    \hline
\end{tabular}}
    \caption{Here $k$ is the number of intersection and complement operators in the regular expression, and $\lambda$ is the number of \emph{critical modules} of the regular expression.} 
    \label{tab:previousresults}
\end{table}

Our main focus is the extended operators (intersection and complement), but we will also consider including the important \emph{interval operator} (also often called the \emph{counting} operator), $R^{[a,b]}$, which specifies the set of strings containing between $a$ and $b$ concatenations of $R$. The interval operator is a standard in most of the practical tools and applications mentioned above for extended operations. While the interval operator has been studied from various perspectives, see e.g.,~\cite{WMN2009, BT2010, KT2003, MS1972, GGM2012, THLSVV2020, HSJTV2023, CL2015}, we are not aware of any published results on the complexity of extended regular expression matching with the interval operator.

\paragraph{Our Contributions} 
The main contributions of this paper are the following. 
\begin{itemize}
    \item A new simple algorithmic framework for time- and space-efficient extended regular expression matching. The framework combines non-deterministic finite automata simulation for (non-extended) regular expression matching, matrix multiplication, and simple tree techniques such as clustering and heavy path decompositions.
    \item We plug in the classic textbook algorithm for (non-extended) regular expression matching by Thompson into our framework to obtain an algorithm for extended regular expression matching that uses $O(n^\omega k + n^2m)$ time and $O(n^2 + m)$ space. This improves both the time and space bounds of all of the previous results for almost all combinations of parameters (see Section~\ref{sec:mainresults} for details). Notably, we significantly improve the space of the previous approaches.     
    Furthermore, the exponent on the $n^\omega$ term in the time complexity is optimal by recent conditional lower bounds. 
 
    \item We plug in more advanced algorithms for (non-extended) regular expression matching in our framework, to improve the $n^2m$ term in the above time bound. This leads to an algorithm using  $O\left(n^\omega k + \frac{n^2m}{\max(w/\log w, \log (n+m))} + m\right)$ time and $O(n^2 + m)$ space. 

    \item We demonstrate the flexibility of our framework by extending it to also handle interval operators in the extended regular expression. This requires a straightforward minor addition to the framework and only incurs a logarithmic overhead in the time complexity. 
    
    \item Along the way, we rephrase the classic dynamic programming algorithm by Hopcroft and Ullman with the straightforward improvement using matrix multiplication.

\end{itemize}

\subsection{Main Results}\label{sec:mainresults}
Throughout the paper, let $m$ and $n$ be the length (number of symbols) of the regular expression $R$ and the string $Q$, respectively. All bounds presented below and in the paper hold for a standard unit-cost RAM with $w$-bit machine words supporting standard arithmetic and logical operations on machine words. This means that the algorithms can be implemented directly in standard imperative programming languages such as C or C++. An index into $R$ and $Q$ can be stored in a single machine word and thus $w \geq \log (n + m)$. 

Our algorithm uses Thompson's classic non-deterministic finite automaton (TNFA)~\cite{Thompson1968} for (non-extended) regular expression matching (see Section~\ref{sec:preliminaries} for the full definition). Given a TNFA $A$ of size $m$, a string $Q$ of length $n$, and two states $s, t$ in $A$, a \emph{TNFA simulation algorithm} outputs every prefix $i$ of $Q$ such that $Q[1,i]$ matches a path from $s$ to $t$. Our main result is the following. 

\begin{theorem}\label{thm:main_reduction}
Let $R$ be an extended regular expression of length $m$ containing $k$ extended operators, let $CS$ be the cluster decomposition of the parse tree of $R$, and let $Q$ be a string of length $n$. Then, we can solve the extended regular expression matching problem for $R$ and $Q$ in space $O(\frac{n^2 \log k}{w} + \max_{C \in CS} S(n,m_C) + m)$ and time 
\[
O\left(n^\omega k + n \sum_{C\in CS}T(n,m_C) + m\right).
\]
Here $T(n,m)$ and $S(n,m)$ denote the time and space, respectively, of a TNFA simulation algorithm.
\end{theorem}
If we plug in Thompson's textbook algorithm~\cite{Thompson1968} for the TNFA simulation, we obtain the following result. 
\begin{corollary}\label{cor:thompson}
Given an extended regular expression $R$ of length $m$ containing $k$ extended operators and a string $Q$ of length $n$, we can solve the extended regular expression matching problem for $R$ and $Q$ in $O(n^2+m)$ space and $O\left(n^\omega k + n^2m\right)$ time.
\end{corollary}
Corollary~\ref{cor:thompson} improves the time and space bounds of all the previous results for almost all combinations of parameters. For the time bound, we either replace the dependency on $m$ with $k$ compared to the dynamic programming solution~\cite{HU1979} (implemented with fast matrix multiplication) or improve the dominant $n^3k/w$ term by a factor $n^{3-\omega} \approx n^{0.6284}$ compared to Yamamoto and Miyazaki~\cite{YM2003}. By the very recent conditional lower bound by Nogami and Terauchi~\cite{NT2026}, the exponent on our dominant $n^\omega$ factor is optimal assuming the $k$-clique hypothesis. For the space bound, we obtain the first $O(n^2 + m)$ space solution, improving the previous results by a factor $m$ or $k/w$. A key component of both the dynamic programming and our algorithms is storing $(n+1)\times (n+1)$ boolean matrices for fast matrix multiplication. Note that this requires $\Omega(n^2)$ \emph{bits} of space, and we match this bound within a logarithmic factor.

With more advanced TNFA simulation algorithms, we can further improve the $O(n^2m)$ term in the time complexity of Corollary~\ref{cor:thompson}. Nearly all of the existing (non-extended) regular expression matching algorithms are TNFA simulation algorithms~\cite{Myers1992, BFC2008, Bille2006, BT2010}.  The state-of-the-art algorithms achieve either $T(n,m) = O(nm/\log (n + m) + n + m)$ and $S(n,m) = O(m + n^\varepsilon)$, for constant $\varepsilon$, $0 < \varepsilon < 1$~\cite{BFC2008} or $T(n,m) = O(nm\log w/w + n + m)$ and $S(n,m) = O(m)$~\cite{Bille2006, BT2010}. Note that the former result is slightly different from what is stated in the original paper (see~\cite{BG2022} for how to improve this). Plugging these into our framework, we obtain the following result.

\begin{theorem}\label{thm:main}
Given an extended regular expression $R$ of length $m$ containing $k$ extended operators and a string $Q$ of length $n$, we can solve the extended regular expression matching problem for $R$ and $Q$ in space $O(\frac{n^2 \log k}{w} + n + m) = O(n^2 +m)$ and time 
\[
O\left(n^\omega k + \frac{n^2m}{\max(w/\log w, \log (n+m))} + m\right)
\]
\end{theorem}

Our algorithmic framework is sufficiently general and simple to allow us to easily and efficiently extend it with other operators. To demonstrate this, we show how to efficiently add support for the interval operator. For an extended regular expression $R$ with interval operators, we achieve the same bounds as in Theorem~\ref{thm:main} except for an additional factor $\log B$ on the running time for the interval operators, where $B$ is the maximum of the upper bounds in the ranges for interval operators in $R$ (see Theorem~\ref{thm:maininterval} for details).

\subsection{Technical Overview}
Our results are based on a surprisingly simple combination of techniques and insights, including a compact representation to store and efficiently combine substring matches, a clustering technique for parse trees of extended regular expressions, and a new efficient combination of finite automaton simulation with the substring match representation to speed up the classic dynamic programming solution.


We use a simple framework, which we call match graphs, for representing substring matches of $R$ in $Q$ as directed graphs. We show how to compactly represent match graphs as boolean matrices and efficiently combine them according to the extended regular expression operators using either boolean matrix multiplication, boolean matrix transitive closure, or simpler operations. 
The match graphs allow us to describe, in a simple way, the classic dynamic programming solution of Hopcroft and Ullman~\cite{HU1979} and its improvement. Essentially, the algorithm processes the parse tree of an extended regular expression in a bottom-up traversal, computing at each node $v$ a match graph representing the substrings of $Q$ that match the subexpression of $R$ represented by $v$. Using the match graph operations, we show how to immediately reduce the time from $O(n^3m)$ to $O(n^\omega m)$.  

To speed up the dynamic programming algorithm, we combine the match graph framework with finite automaton simulation. We first decompose the parse tree of $R$ into a hierarchy of connected subtrees, called \emph{clusters}, such that each leaf cluster contains no node corresponding to an extended operator, and each internal cluster contains exactly one node that either corresponds to an extended operator or is the lowest common ancestor of a pair of extended operators, thus guaranteeing that we have no more than $O(k)$ clusters, where $k$ is the number of extended operators in $R$. We then process the hierarchy of clusters in a bottom-up traversal and compute at each cluster $C$ the match graph representation for the subtree of the parse tree rooted at the root of $C$. To achieve our desired running time, the key technical challenge is to efficiently combine the slow match graph computation with fast finite automaton simulation of a (non-extended) regular expression corresponding to $C$. We present a new technique that achieves this using only a constant number of match graph operations and a constant number of finite automata simulations. In combination with the textbook algorithm by Thompson~\cite{Thompson1968} for the finite automata simulations, this leads to an algorithm for extended regular expression matching using $O(n^\omega k + n^2 m)$ time and $O(n^2k/w + n + m)$ space. 

We then improve the space to $O(n^2\log k/w + n + m) = O(n^2 + m)$. To do so, we show how to recursively process the hierarchy of clusters top-down according to the number of descendants of the children of each cluster and carefully discard match graphs along the way. This gives us a simple algorithm using $O(n^\omega k + n^2 m)$ time and $O(n^2 + m)$ space. 

To achieve the full result of Theorem~\ref{thm:main}, we generalize the finite automaton simulations in our algorithm to work with a black-box implementation of a (non-extended) regular expression matching algorithm. 
Plugging in current state-of-the-art results improves the $O(n^2 m)$ term in the running time by a factor $\max(w/\log w, \log (n+m))$ and thus leads to Theorem~\ref{thm:main}. 

Finally, to demonstrate the power and simplicity of our approach, we extend our algorithm to also include interval operators. We do so by defining and efficiently implementing another simple match graph operation for the interval operator and plugging it into our framework.




\subsection{Outline}
We review regular expressions and finite automata in Section~\ref{sec:preliminaries}. We present match graphs and the improved implementation of the Hopcroft and Ullman classic dynamic programming algorithm in Section~\ref{sec:dynamicprogramming}. Section~\ref{sec:improvedalgorithm} presents (a simplified version of) our new algorithm, together with our clustering and a generalization of match graphs. In Sections~\ref{sec:reducingspace} and~\ref{sec:blackbox}, we show how to reduce the space and implement our algorithm with a black-box finite automaton simulation. Finally, in Section~\ref{sec:intervaloperator}, we show how to extend our results to efficiently support the interval operator.

\section{Regular Expressions and Finite Automata}\label{sec:preliminaries}
We review the classical concepts used in the paper. For more details, see e.g., Hopcroft and Ullman~\cite{HU1979}.  

\paragraph{Regular Expressions} 
We consider the set of non-empty extended regular expressions over an alphabet $\Sigma$, defined recursively as follows. If $\alpha \in \Sigma \cup \{\epsilon\}$, then $\alpha$ is an extended regular expression, and if $S$ and $T$ are extended regular expressions, then so is the \emph{concatenation} ($S \odot T$), the \emph{union} ($S|T$), the \emph{intersection} ($S \cap T$), the \emph{complement} ($\neg S$), and the \emph{star} ($S^*$). We often omit the concatenation $\odot$ when writing extended regular expressions. 

The \emph{language} $L(R)$ generated by an extended regular expression $R$ is defined recursively as follows. If $\alpha \in \Sigma \cup \{\epsilon\}$, then $L(\alpha)$ is the set containing the single string $\alpha$.  If $S$ and $T$ are extended regular expressions, then $L(S \odot T) = L(S)\odot L(T)$, that is, any string formed by the concatenation of a string in $L(S)$ with a string in $L(T)$, $L(S | T) = L(S) \cup L(T)$, $L(S \cap T) = L(S) \cap L(T)$, $L(\neg{S}) = \Sigma^* - L(S)$, and $L(S^*) = \bigcup_{i \geq 0} L(S)^i$, where $L(S)^0 = \{\epsilon\}$ and $L(S)^i = L(S)^{i-1} \odot L(S)$, for $i > 0$. Given an extended regular expression $R$ and a string $Q$, the \emph{extended regular expression matching problem} is to decide if $Q \in L(R)$.  

We call the intersection and the complement operators the \emph{extended operators}. Without these, the expression is simply called a \emph{regular expression} and the corresponding matching problem is the \emph{regular expression matching problem}.

\begin{figure}[t]
    \centering
    \includegraphics[]{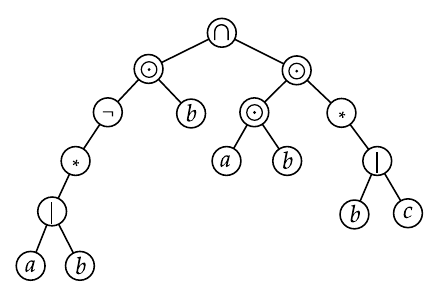}
    \caption{Parse tree for the extended regular expression $(\neg((a|b)^*) b)\cap (a b(b|c)^*)$.}
    \label{fig:parse_tree}
\end{figure}

The \emph{parse tree} $T(R)$ of $R$ (not to be confused with the parse of $Q$ wrt. $R$) is a rooted, binary tree representing the hierarchical structure of $R$.  The leaves of $T(R)$ are labeled by a character from $\Sigma \cup \{\epsilon\}$ and internal nodes are labeled by either $\odot$, $\mid$, $\cap$, $\neg$, or $^*$. For simplicity in our presentation, we identify extended regular expressions with their parse tree. We denote the subtree (or equivalently the subexpression) rooted at $v$ by $R(v)$. See Figure~\ref{fig:parse_tree} for an example. Given two nodes $v$ and $w$ in $T(R)$, the \emph{lowest common ancestor}, denoted $\lca(v,w)$, is the common ancestor of $v$ and $w$ with the largest depth.

\paragraph{Finite Automata}
A \emph{finite automaton} is a tuple $A = (V, E, \Sigma, \Theta, \Phi)$, where $V$ is a set of nodes called \emph{states}, $E \subseteq (V \times V \times \Sigma \cup \{\epsilon\})$ is a set of directed edges between states called \emph{transitions} each labeled by a character from $\Sigma \cup \{\epsilon\}$, $\Theta \subseteq V$ is a set of \emph{start states}, and $\Phi \subseteq V$ is a set of \emph{accepting states}. In short, $A$ is an edge-labeled directed graph with designated subsets of start and accepting nodes. $A$ is a \emph{deterministic finite automaton} (DFA) if $A$ does not contain any $\epsilon$-transitions,  all outgoing transitions of any state have different labels, and there is exactly one start state. Otherwise, $A$ is a \emph{nondeterministic finite automaton} (NFA).  

Given a string $Q$ and a path $p$ in $A$, we say that $p$ and $Q$ match if the concatenation of the labels on the transitions in $p$ is $Q$. Given a state $s$ in $A$ and a character $\alpha$ we define the \emph{state-set transition} $\delta_A(s, \alpha)$ to be the set of states reachable from $s$ through paths matching $\alpha$ (note that the paths may include transitions labeled $\epsilon$). For a set of states $S$ we define $\delta_A(S,\alpha) = \bigcup_{s\in S} \delta_A(s,\alpha)$. We say that $A$ \emph{accepts} a string $Q$ if there is a path from a state in $\Theta$ to a state in $\Phi$ that matches $Q$. Otherwise, $A$ \emph{rejects} $Q$. We can use a sequence of state-set transitions to test acceptance of a string $Q$ of length~$n$ by computing a sequence of state-sets $S_0, \ldots, S_n$, given by $S_0 = \delta_A(\Theta, \epsilon)$ and $S_i = \delta_A(S_{i-1}, Q[i])$, $i=1, \ldots, n$. We have that $\Phi \cap S_n \neq \emptyset$ iff $A$ accepts $Q$.

\paragraph{Thompson NFA}
\begin{figure}[t]
    \centering
    \includegraphics[width=0.7\linewidth]{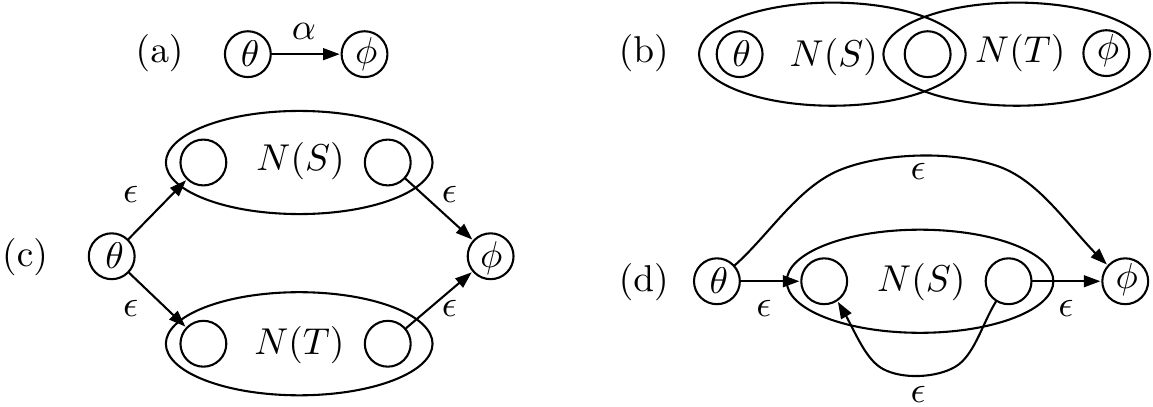}
    \caption{Thompson's recursive NFA construction. The regular
   expression $\alpha \in \Sigma \cup \{\epsilon\}$ corresponds to NFA
   $(a)$. If $S$ and $T$ are regular expressions then $N(ST)$,
   $N(S|T)$, and $N(S^*)$ correspond to NFAs $(b)$, $(c)$, and $(d)$,
   respectively.  In each of these figures, the leftmost node $\theta$
   and rightmost node $\phi$ are the start and the accept nodes,
   respectively.  For the top recursive calls, these are the start and
   accept nodes of the overall automaton. In the recursions indicated,
   e.g., for $N(ST)$ in (b), we take the start node of the
   subautomaton $N(S)$ and identify with the state immediately to the
  left of $N(S)$ in (b). Similarly the accept node of $N(S)$ is
   identified with the state immediately to the right of $N(S)$ in
   (b).}
    \label{fig:Thompson}
\end{figure}

Given a (non-extended) regular expression $R$, we can construct an NFA accepting precisely the strings in $L(R)$ by several classical  methods~\cite{MY1960, Glushkov1961, Thompson1968}. In particular, Thompson~\cite{Thompson1968} gave the simple textbook construction shown in Figure~\ref{fig:Thompson}. We will call an NFA constructed with these rules a \emph{Thompson NFA} (TNFA). A TNFA $N(R)$ for $R$ has at most $2m$ states, at most $4m$ transitions, and can be computed in $O(m)$ time.  We can compute a state-set transition on a single character in $O(m)$ time with a breadth-first search. Hence, it follows that we can solve (non-extended) regular expression matching on $Q$ in $O(nm)$ time and $O(m)$ space using the above state-set transition algorithm.

\section{Dynamic Programming Algorithm}\label{sec:dynamicprogramming}
Let $R$ be an extended regular expression of length $m$ and let $Q$ be a string of length $n$.
We review Hopcroft and Ullman's dynamic programming algorithm~\cite{HU1979} and show how to implement it using fast matrix multiplication in $O(n^\omega m)$ time and $O(n^2m/w + n + m)$ space.

\begin{figure}[t]
    \centering
    \includegraphics[width=0.95\linewidth]{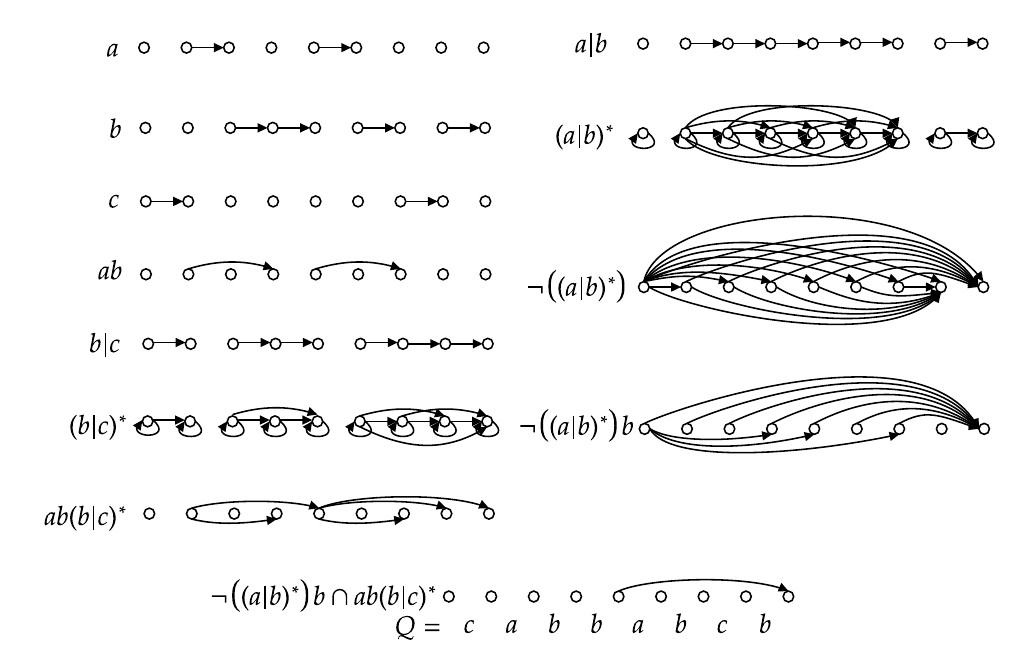}
    \caption{Example of the dynamic programming algorithm for $Q=cabbabcb$ and $R=(\neg((a|b)^*) b)\cap (a b(b|c)^*)$. The parse tree of $R$ is shown in Figure~\ref{fig:parse_tree}. The match graph for each node $v$ in the parse tree is labeled by the subexpression $R(v)$. The match graph of the root node of the parse tree (bottom) tells us that the only matching substring in $Q$ is $Q[5,8] = abcb$.}
    \label{fig:dyn_prog}
\end{figure}

\subsection{Match Graphs}\label{sec:matchgraphs}
We first define a simple, natural graph representation, called a \emph{match graph}, to represent matching substrings of a subexpression of an extended regular expression.

Given a subexpression $R(v)$ rooted in node $v$ in $R$, we define the \emph{match graph} for $R(v)$, denoted $G(v)$, to be the directed graph on $n+1$ vertices $g_1, \ldots, g_{n+1}$  such that there is an edge from $g_i$ to $g_j$ iff $R(v)$ matches the substring $Q[i,j-1]$, $1 \leq i \leq j \leq n+1$. Note that $G(v)$ is acyclic except for self-loops. Also, if $R(v)$ matches the empty string $\epsilon$, then every vertex in $G(v)$ has a self-loop. See Figure~\ref{fig:dyn_prog} for an example.

\subsection{Algorithm}\label{sec:DP_algo}
We compute the match graphs for each subexpression of $R$ in a bottom-up traversal of the parse tree. At each node $v$, we construct a match graph $G(v)$ on $n+1$ vertices $g_1, \ldots, g_{n+1}$ and add edges according to the following cases:

\begin{description}
    \item[Case 1: $v$ is a leaf labeled $\alpha \in \Sigma$.] We add an edge $(g_i, g_{i+1})$ to $G(v)$ for each position $i$ in $Q$ such that $Q[i]=\alpha$ 
    \item[Case 2: $v$ is an internal node with two children.] Let $v_1$ and $v_2$ be the left and right children of $v$, respectively. There are three subcases: 
    \begin{description}
        \item[Case 2.1: $v$ is labeled $\odot$.] We add an edge $(g_i, g_j)$ to $G(v)$ for each pair $i, j$ such that there exist edges $(g_i, g_l)$ in $G(v_1)$ and $(g_l, g_j)$ in $G(v_2)$.   
        \item[Case 2.2: $v$ is labeled $|$.] We add an edge $(g_i, g_j)$ to $G(v)$ for each pair $i,j$ if $(g_i, g_j)$ in $G(v_1)$ or $(g_i, g_j)$ in $G(v_2)$.
        \item[Case 2.3: $v$ is labeled $\cap$.] We add an edge $(g_i, g_j)$ to $G(v)$ for each pair $i,j$ if $(g_i, g_j)$ in $G(v_1)$ and $(g_i, g_j)$ in $G(v_2)$. 
    \end{description}
    \item[Case 3: $v$ is an internal node with a single child.] Let $v_1$ be the child of $v$. There are two subcases: 
    \begin{description}
        \item[Case 3.1: $v$ is labeled $\neg$.] We add an edge $(g_i, g_j)$ to $G(v)$ for each edge not in $G(v_1)$ (including self-loops). 
        \item[Case 3.2: $v$ is labeled $\star$.] We add an edge $(g_i, g_j)$ to $G(v)$ for each pair $i,j$ with a directed path from $g_i$ to $g_j$ in $G(v_1)$ and add self-loops to all vertices, i.e., $G(v)$ is the transitive closure of $G(v_1)$.
    
    \end{description}
\end{description}
Finally, we report that $R$ matches $Q$ iff the edge $(g_1, g_{n+1})$ is in $G(r)$, where $r$ is the root of $T(R)$. 

By induction on the parse tree, it follows that the algorithm correctly determines if $R$ matches $Q$. To implement the algorithm efficiently, we store each match graph $G(v)$ as a boolean  $(n+1) \times (n+1)$ matrix, such that entry $(i,j)$ is $1$ iff there is an edge $(g_i, g_j)$ in $G(v)$. Consider the cases in the bottom-up traversal. We have that case 2.1 is a single boolean matrix multiplication of the matrices representing $G(v_1)$ and $G(v_2)$, and case 3.2 is the \emph{transitive closure} of the matrix representing $G(v_1)$, which can also be computed in the same time as a single matrix multiplication~\cite{FM1971}. The remaining cases are straightforward to implement in $O(n^2)$ time. Thus, we use at most $O(n^\omega)$ time on each node in the parse tree, and hence the total time is $O(n^\omega m)$. Since each matrix uses $O(n^2)$ \emph{bits}, the total space is $O(n^2m/w + n + m)$.

\section{The Improved Algorithm}\label{sec:improvedalgorithm}
We now present an extended regular expression matching algorithm that uses $O(n^\omega k + n^2m)$ time and $O(n^2 k/w + n + m)$ space. For simplicity we describe the algorithm using simple state-set transitions on the TNFAs. We will show how to improve the space and generalize it to use black-box implementations of fast (non-extended) regular expression matching algorithms in the following sections.

\subsection{Clustering}\label{sec:clustering}
Let $R$ be an extended regular expression of length $m$ containing $k$ extended operators. We will use a clustering of the parse tree $T(R)$ into $O(k)$ node-disjoint clusters such that each cluster contains at most one node labeled with an extended operator. While clustering techniques (also called \emph{modules}) are also in previous solutions~\cite{Yamamoto2001, YM2003}, our use of clustering in combination with match graphs and finite automaton simulation is new and differs significantly from previous approaches.

We define the set of \emph{extended nodes} $P$ to be the subset of nodes $v$ in $T(R)$ such that (i) $v$ is labeled by an extended operator or (ii) $v$ is the lowest common ancestor of two nodes that are labeled by an extended operator. The edges in $T(R)$ from a node $p$ in $P$ to any children of $p$ are \emph{external edges} and all other edges are \emph{internal edges}. Deleting all external edges partitions $T(R)$ into node-disjoint subtrees that we call \emph{clusters}, and the set of all these clusters is the \emph{cluster partition} $CS$ of $T(R)$. Contracting all internal edges induces a \emph{macro tree}, where each cluster is represented by a single node. We will use standard tree terminology on the macro tree and refer to clusters as child clusters, internal clusters, parent clusters, leaf clusters, etc. We can construct the clustering in a bottom-up traversal of $T(R)$ in $O(m)$ time. See Figure~\ref{fig:clustered_tree} for an illustration.

\begin{figure}[t]
    \centering
    \includegraphics[width=0.95\linewidth]{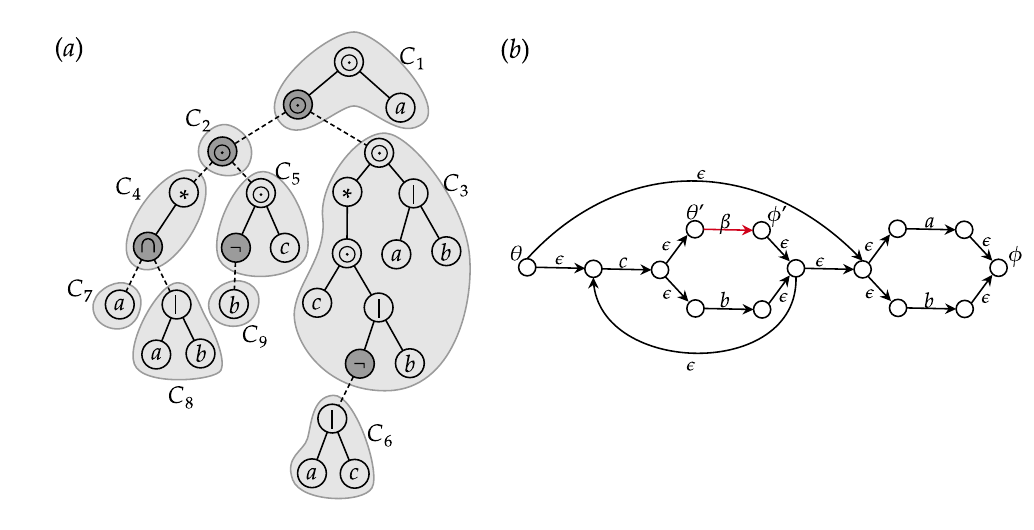}
    \caption{$(a)$ The clustering of the parse tree of an extended regular expression. The dark gray nodes are the nodes from $P$. The dashed lines are external edges and the solid lines are internal edges. $(b)$ The automaton $A_{C_3}$ corresponding to cluster $C_3$ in $(a)$.}
    \label{fig:clustered_tree}
\end{figure}

\begin{lemma}\label{lem:clustering}
        Let $R$ be a regular expression with $k$ extended operators. The cluster partition $CS$ of the parse tree of $R$ satisfies the following properties:
        \begin{itemize}
            \item[(i)] a leaf cluster in $CS$ contains no extended nodes, 
            \item[(ii)] each internal cluster in $CS$ contains exactly one extended node, and 
            \item[(iii)] the total number of clusters in $CS$ is $O(k)$. 
        \end{itemize}
\end{lemma}
\begin{proof}
Let $P$ be the set of extended nodes in $T(R)$ and let $CS$ be the corresponding cluster partition. (i) By definition, a leaf cluster in $CS$ cannot contain an extended  node $p \in P$. (ii) Consider an internal cluster $C \in CS$. By definition, an internal cluster must contain at least one extended node. Note that an extended node cannot be in the same cluster as any of its descendants. Assume for contradiction that $C$ contains two distinct extended nodes $p_1, p_2 \in P$. Then, the lowest common ancestor $p_3 = \lca(p_1, p_2)$ is an extended node also in $C$ (if $p_3$ was not in $C$, then both $p_1$ and $p_2$ can't be in $C$). If $p_3 \neq p_1$, $p_3$ is a proper ancestor of $p_1$ and thus $p_3$ cannot be in $C$. Symmetrically, if  $p_3 \neq p_2$, $p_3$ can also not be in $C$. 
(iii) Consider the macro tree $M$. Then, $M$ has at most $2k$ leaf clusters, since the parents of the roots of the leaf clusters are all labeled by an extended operator. There can be at most $k-1$ extended nodes that do not correspond to extended operators. It follows that there are at most $2k-1$ internal clusters. In total, we have at most $4k-1$ clusters. 
\end{proof}

Let $CS$ be a cluster partition with extended nodes $P$ and let $C \in CS$ be a cluster with more than one node and root $v$. We define a corresponding TNFA, denoted $A_C$, as follows. If $C$ is a leaf cluster, then $A_C$ is the TNFA of the corresponding subexpression $R(v)$. If $C$ is an internal cluster, we construct the TNFA corresponding to the parse tree defined by $C$, where we have replaced the single extended node $p \in P$ by a special character $\beta \not\in \Sigma$. The extended node $p$ corresponds to an \emph{extended start state} and \emph{extended end state}, denoted $\theta'$ and $\phi'$, respectively, in $A_C$. The transition  from $\theta'$ to $\phi'$ labeled $\beta$ is called the \emph{extended transition}. See Figure~\ref{fig:clustered_tree}(b). Note that by Thompson's construction, $A_C$ has $O(|C|)$ states and transitions. 

Intuitively, the extended transition represents the extended regular expression $R(p)$.
Let $Q$ be a string that matches the regular expression $R(v)$. Then, there is a decomposition of $Q$ into substrings $Q= Q_1 \odot\tilde{Q}_1\odot Q_2 \odot\tilde{Q}_2 \odot\cdots \odot Q_{\ell-1} \odot\tilde{Q}_{\ell-1}\odot Q_\ell$, such that $\tilde{Q}_i$ matches the extended regular expression $R(p)$ and $Q_i$ matches parts of the non-extended part of the regular expression $R(v)$ for all $i$. If we replace each $\tilde{Q}_i$ with $\beta$,  we obtain a string $Q'= Q_1 \odot\beta\odot Q_2 \odot\beta \odot\cdots \odot\beta\odot Q_\ell$ that is accepted by $A_C$.  

\begin{example}\label{ex:substringdecomp} Let $v$ be the root of cluster $C_3$ in Figure~\ref{fig:clustered_tree}(a) and let $Q=cabcbba$.  There are several ways to match $Q$ to the regular expression $R(v)$. One such way is to match the two substrings $Q[2,3] = ab$ and $Q[5,6] = bb$ to $R(p) = \neg (a \mid c)$. Replacing these two substrings with $\beta$ we get $Q' = c\beta c \beta a$ which is accepted by the automaton in Figure~\ref{fig:clustered_tree}(b).
\end{example}

\subsection{General Match Graphs}\label{sec:matchgraphsgeneral}
In this section we generalize our match graphs from Section~\ref{sec:matchgraphs} and give efficient operations on them to implement our algorithm. Given a TNFA $A$ and a pair of states $(s, t)$, the \emph{match graph} for $A$, $s$, and $t$, denoted $G(A,s,t)$, is the directed graph on $n+1$ vertices $g_1, \ldots g_{n+1}$ such that there is an edge from $g_i$ to $g_j$ iff there is a path from $s$ to $t$ in $A$ matching $Q[i,j-1]$, $1 \leq i \leq j \leq n+1$. As before, if there is a path from $s$ to $t$ matching the empty string, every vertex in $G(A,s,t)$ has a self-loop. If $s=t$, then every vertex in $G(A,s,t)$ also has a self-loop.

Given match graphs $G$ and $F$, we define match graphs combining $G$ and $F$ corresponding to the extended regular expression operators, namely the match graphs $G\odot F$, $G\mid F$, $G\cap F$, $\neg G$, and $G^\ast$. All of these are on the set of vertices $g_1, \ldots, g_{n+1}$. We define the edges as follows. 
\begin{itemize}
    \item For $G \odot F$, we add an edge $(g_i, g_j)$ for each pair $i, j$ such that there exist edges $(g_i, g_l)$ in $G$ and $(g_l, g_j)$ in $F$.
    \item For $G \mid F$, we add an edge $(g_i, g_j)$  for each pair $i,j$ if $(g_i, g_j)$ in $G$ or $(g_i, g_j)$ in $F$.
    \item For $G \cap F$, we add an edge $(g_i, g_j)$ for each pair $i,j$ if $(g_i, g_j)$ in $G$ and $(g_i, g_j)$ in $F$.
    \item For $\neg G$, we add an edge $(g_i, g_j)$ for each edge not in $G$.
    \item For $G^*$, we add an edge $(g_i, g_j)$ for each pair $i,j$ with a directed path from $g_i$ to $g_j$ in $G$ and add self-loops to all vertices.
\end{itemize}

As in Section~\ref{sec:dynamicprogramming}, we can represent match graphs as boolean matrices and implement the operations efficiently. 
\begin{lemma}\label{lem:matchgraphoperations}
    Let $G$ and $F$ be match graphs on a string of length $n$. Then, we can store them in $O(n^2/w + 1) $ space and compute
    \begin{itemize}
    \item[(i)] $G \odot F$ and $G^\ast$ in $O(n^\omega)$ time, and 
    \item[(ii)] $G \mid F$, $G \cap F$, and $\neg G$ in $O(n^2)$ time.
    \end{itemize}
\end{lemma}
\begin{proof}
As in the dynamic programming algorithm from Section~\ref{sec:dynamicprogramming}, we store a match graph as a boolean matrix in $O(n^2)$ bits. (i) The operation $G \odot F$ is a boolean matrix multiplication and $G^\ast$ is the transitive closure of a matrix. Both use $O(n^\omega)$ time. (ii) $G \mid F$, $G \cap F$, and $\neg G$ are straightforward to implement in constant time for each entry.  
\end{proof}
We will also need to efficiently construct match graphs for TNFA. 
\begin{lemma}\label{lem:tnfamatchgraph}
    Let $A$ be a TNFA of size $m$. Given a pair of states $(s,t)$, we can construct the match graph $G(A, s, t)$ for a string of length $n$ in $O(n^2 m)$ time and $O(n^2/w + m)$ space.  
\end{lemma}
\begin{proof}
    Let $Q$ be a string of length $n$. For each suffix $Q[i, n]$ of $Q$, we perform a sequence of state-set transitions to find all substrings $Q[i,j]$ that match a path from $s$ to $t$. To do so, we compute a sequence of state-sets $S_i,\ldots, S_n$,  such that $S_i = \delta_A(s, Q[i])$ and $S_{h} = \delta_A(S_{h-1}, Q[h])$, for $h = i+1,\ldots, n$. We have that the state-set $S_h$ contains state $t$ iff there is a path from $s$ to $t$ matching $Q[i,h]$ and thus this finds all matching substrings starting at position $i$. We can thus construct the corresponding match graph from this procedure.  
    
    Since a state-set transition on a single character takes $O(m)$ time, we process suffix $i$ in $O((n-i+1)m)$ time. In total, we use $O(n^2 m)$ time. Since we only need to store the match graph and a constant number of state-sets during the algorithm, we use $O(n^2/w + m)$ space.  
\end{proof}

\subsection{Algorithm}\label{sec:improvedalgorithmsub}
We are now ready to describe the improved algorithm. 
We first construct the parse tree $T(R)$ and the clustering according to Lemma~\ref{lem:clustering}. To speed up the dynamic programming algorithm of Section~\ref{sec:dynamicprogramming}, we will show how to efficiently compute the match graph of the root of each cluster in a bottom-up traversal of the macro tree. 

\paragraph{High level idea} The main idea is the following. As described in Section~\ref{sec:clustering} if a string $Q$ matches $R(v)$ then $Q$ can be partitioned into substrings $Q= Q_1 \odot\tilde{Q}_1\odot Q_2 \odot\tilde{Q}_2 \odot\cdots \odot Q_{\ell-1} \odot\tilde{Q}_{\ell-1}\odot Q_\ell$, such that $\tilde{Q}_i$ matches $R(p)$ for all $i$, and the string $Q' = Q_1\odot \beta \odot Q_2\odot \beta \odot \cdots \odot Q_{\ell-1} \odot \beta \odot Q_\ell$ is accepted by $A_C$. Furthermore, the string $Q_1$ must match a path from $\theta$ to $\theta'$ in $A_C$, $Q_\ell$ must match a path from $\phi'$ to $\phi$, and for all other~$i$, $Q_i$ must match a path from $\phi'$ to $\theta'$. 

\begin{example}\label{ex:paths}
     Consider the example from Example~\ref{ex:substringdecomp} and the automaton from Figure~\ref{fig:clustered_tree}(b). Here $Q'=c\beta c\beta a$ and $Q_1 = c$ matches a path from $\theta$ to $\theta'$, $Q_2 = c$ matches a path from $\phi'$ to $\theta'$, and $Q_3 = a$ matches a path from $\phi'$ to $\phi$.
\end{example}
We therefore compute the match graphs $G(A_C,\theta,\theta'), G(A_C,\phi',\theta')$, and  $G(A_C, \phi',\phi)$. 
To include the case where the parsing of $Q$ does not use the subexpression $R(p)$ we also compute the match graph $G(A_C, \theta,\phi)$. 
We then combine these match graphs with the match graph $G(p)$ using the match graph operations from the previous section which allows us to write the combinations of the match graphs as a regular expression of match graphs (see Figure~\ref{fig:paths_in_A_C}(a)). The intuition is that there  will be an edge $(l,r)$ in the match graph $G(v)$  if  $(l,r)$ is an edge in $G(A_C, \theta,\phi)$,  or if there exists a sequence of edges 
$(l_1,r_1),(\tilde{l}_1, \tilde{r}_1), (l_2,r_2),(\tilde{l}_2, \tilde{r}_2),\ldots, (\tilde{l}_{\ell-1}, \tilde{r}_{\ell-1}),(l_\ell,r_\ell)$
such that  $l_1 = l$, $r_\ell = r$, $\tilde{l}_i = r_i$ for all $i$,  and $l_i = \tilde{r}_{i-1}$ for all  $i \neq 1$. Furthermore, $(l_1,r_1) \in G(A_C,\theta,\theta')$, $(l_{\ell}, r_\ell) \in G(A_C, \phi',\phi)$, $(\tilde{l}_i, \tilde{r}_i) \in G(p)$ and $(l_i, r_i) \in G(A_C,\phi',\theta')$ for $i \not\in \{1,\ell\}$. See Figure~\ref{fig:paths_in_A_C}(b) for an example.

\paragraph{Algorithm} We now present the full algorithm. First, construct the parse tree $T(R)$,  the clustering according to Lemma~\ref{lem:clustering}, the corresponding macro tree $M$ of the clusters $CS$, and the set of extended nodes $P$. We traverse the macro tree bottom-up. At each cluster $C$ with root $v$, we compute the match graph $G(v)$ as follows. If $C$ is a leaf cluster, $C$ contains no extended node. We directly construct the TNFA $A_C$ and compute the match graph $G(v) = G(A_C, \theta, \phi)$ according to Lemma~\ref{lem:tnfamatchgraph}.  Otherwise, $C$ is an internal cluster with a single extended node $p \in P$. We perform the following steps.



\paragraph{Step 1: Compute $G(p)$.} We first compute $G(p)$. Since we compute the match graphs of the clusters bottom up we have already computed the match graphs of $p$'s children in the parse tree, as these will be the roots of the child clusters of $C$ in the macro tree. 
There are two cases. 
            \begin{description}
                \item[Case (i): $p$ has a single child $u$.]
                From the clustering, it follows that $p$ is labeled $\neg$. We compute $G(p)$ from $G(u)$ as described in Lemma~\ref{lem:matchgraphoperations}.
                \item[Case (ii): $p$ has two children $u$ and $w$.] We compute $G(p)$ by composing the match graphs of $G(u)$ and $G(w)$ according to the label of $p$ as described in Lemma~\ref{lem:matchgraphoperations}. 
            \end{description}
%
%
If $v = p$ ($C$ is a single-node cluster), then $G(v) = G(p)$ and we are done.

\begin{description}
        \item[Step 2:] We construct the TNFA $A_C$ with extended start state $\theta'$ and extended end state $\phi'$ and compute the following match graphs 
        \begin{align*}
            G_{\theta,\phi} &= G(A_C, \theta,\phi) \\
            G_{\theta,\theta'} &= G(A_C,\theta,\theta') \\
            G_{\phi',\theta'} &= G(A_C,\phi',\theta') \\
            G_{\phi',\phi} &= G(A_C, \phi',\phi)
        \end{align*}
                    
    \item[Step 3:] Compute and return the match graph $$G(v)= G_{\theta,\phi}|\left(G_{\theta,\theta'} \odot G(p) \odot ( G_{\phi',\theta'}\odot G(p))^* \odot  G_{\phi',\phi}  \right)$$  
\end{description}

\begin{figure}[t]
    \centering
    \includegraphics[scale = 0.8]{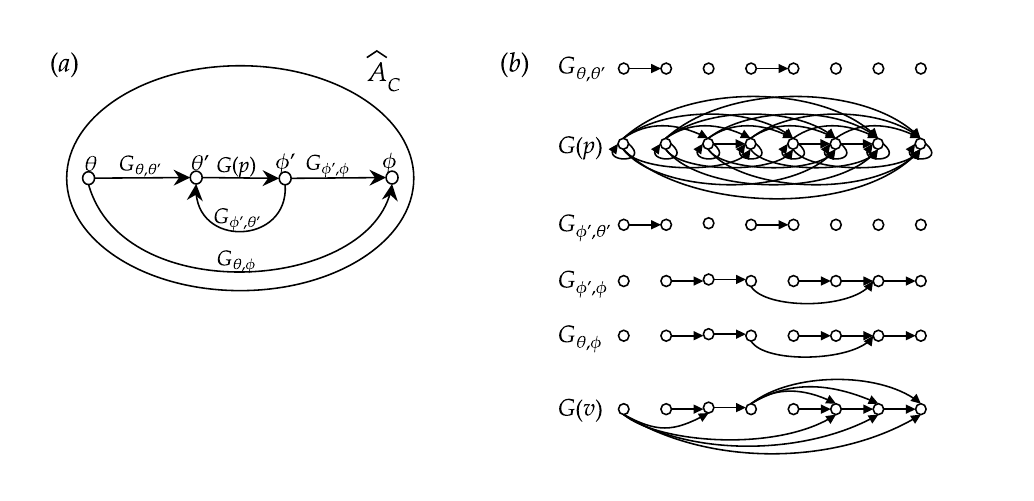}
    \caption{$(a)$ An illustration of the matching paths represented by the five match graphs used in step 3 of the algorithm. $(b)$ The match graphs computed for cluster $C_3$ in Figure~\ref{fig:clustered_tree} and the string $Q= cabcbba$ from Example~\ref{ex:substringdecomp} and~\ref{ex:paths}.}
    \label{fig:paths_in_A_C}
\end{figure}

\subsection{Correctness}
Let $G$ be the match graph for cluster $C$ returned by the algorithm. We will prove that $G=G(v)$.

If $C$ is a leaf cluster, then the correctness follows immediately from the correctness of the state-set transition algorithm in Lemma~\ref{lem:tnfamatchgraph}. Let $C$ be an internal cluster and assume by induction that we have correctly computed the match graphs of the child clusters of $C$. It then follows from Lemma~\ref{lem:matchgraphoperations} that $G(p)$ is the match graph of $R(p)$ (Step 1). If $C$ only contains a single node $v$, then $v = p$ and the correctness follows. Assume $|C|> 1$. The correctness of the match graphs in Step 2 follows from Lemma~\ref{lem:tnfamatchgraph}. 
That $G=G(v)$ now follows from the following lemma.


\begin{lemma}\label{lemma_correcness}
    Let $C$ be a cluster with more than one node, and root node $v$. Then $G=G_{\theta,\phi}|(G_{\theta,\theta'} \odot G(p) \odot (G_{\phi',\theta'}\odot G(p))^* \odot  G_{\phi',\phi}  )$ is the match graph of the cluster $C$, that is, $G=G(v)$.
\end{lemma}
\begin{proof}
We will show that  $(g_l,g_r) \in G$ if and only if $Q[l,r-1]$ matches $R(v)$. We first define an automaton $\widehat{A}_C$ used only for the proof. 
The automaton $\widehat{A}_C$ contains the same states and transitions as the TNFA $A_C$, in addition to the following paths from $\theta'$ to $\phi'$: For each interval $[i,j-1]$ of $Q$ that matches $R(p)$, we add a path from $\theta'$ to $\phi'$ with transition labels concatenating to $Q[i,j-1]$. Thus, $\widehat{A}_C$ accepts a string from the alphabet $\Sigma$ if and only if it matches $R(v)$. Recall that the special character $\beta \notin \Sigma$, so the extended transition is never used.
Note, that $A_C$ is a subgraph of $\widehat{A}_C$.

First we show that $(g_l,g_r)\in G$ implies that $Q[l,r-1]$ matches $R(v)$.
Assume $(g_l,g_r)\in G$. Then $(g_l,g_r)$ is either in $G_{\theta,\phi}$, in $G_{\theta,\theta'}\odot G(p)\odot (G_{\phi',\theta'}\odot G(p))^* \odot G_{\phi',\phi}$ or in both. 
%
Suppose first that $(g_l,g_r) \in G_{\theta,\phi}$. Then there is a path in $A_C$ from $\theta$ to $\phi$ matching $Q[l,r-1]$, so $Q[l,r-1]$ matches $R(v)$. 
%
Otherwise, we have that $(g_l,g_r)\in G_{\theta,\theta'}\odot G(p)\odot (G_{\phi',\theta'}\odot G(p))^* \odot G_{\phi',\phi}$. 
Then there is a sequence of edges 
\[(g_{l_1},g_{r_1}),(g_{\tilde{l}_1}, g_{\tilde{r}_1}), (g_{l_2},g_{r_2}),(g_{\tilde{l}_2}, g_{\tilde{r}_2}),\ldots, (g_{\tilde{l}_{\ell-1}}, g_{\tilde{r}_{\ell-1}}),(g_{l_\ell}g_{r_\ell})\]
where  $l_1 = l$, $r_\ell = r$, $\tilde{l}_i = r_i$ for all $i$,  and $l_i = \tilde{r}_{i-1}$ for all  $i \neq 1$, such that $(g_{l_1},g_{j_1}) \in G(A_C,\theta,\theta')$, $(g_{l_{\ell}}, g_{r_\ell}) \in G(A_C, \phi',\phi)$, $(g_{\tilde{l}_i}, g_{\tilde{r}_i}) \in G(p)$ and $(g_{l_i}, g_{r_i}) \in G(A_C,\phi',\theta')$ for $i \not\in \{1,\ell\}$.
The edge $(g_{l_1},g_{r_1})$ corresponds to a path in $A_C$ matching $Q[l_1,r_1-1] = Q[l, r_1-1]$ and the edge $(g_{l_\ell},g_{r_\ell})$ corresponds to a path in $A_C$ matching $Q[l_\ell,r_\ell-1] = Q[l_\ell, r-1]$.
Each of the edges $(g_{l_i},g_{r_i})$, for $1< i < \ell$, corresponds to a path in $A_C$ matching $Q[l_i,r_i-1]$, and each of the edges $(g_{\tilde{l}_i},g_{\tilde{r}_i})$, for $1\leq i < \ell$, corresponds to a path in $\widehat{A}_C$ from $\theta'$ to $\phi'$ matching $Q[\tilde{l}_i,\tilde{r}_i-1]$. 
Note, that any of the edges $(g_{l_i},g_{r_i})$ or $(g_{\tilde{l}_i},g_{\tilde{r}_i})$ could be loops, in which case the labels on the corresponding path only contains transitions labeled $\epsilon$.
The concatenation of these paths is a path from $\theta$ to $\phi$ in $\widehat{A}_C$, that matches $Q[l,r-1]$. Thus $Q[l,r-1]$ matches $R(v)$.


Next, we show  that if $Q[l,r-1]$ matches $R(v)$ then $(g_l,g_r) \in G$. Let $\Gamma$ be the matching path of 
$Q[l,r-1]$ from $\theta$ to $\phi$ in $\widehat{A}_C$.
Suppose that $\theta'\not\in \Gamma$. Then $\Gamma$ is fully contained in $A_C$, and thus $(g_l,g_r)\in G_{\theta,\phi}$ by the construction of $G_{\theta,\phi}$. 
%
Otherwise, we decompose $\Gamma$, splitting it into subpaths \[\Gamma_1,\tilde{\Gamma}_1, \Gamma_2, \tilde{\Gamma}_2,\ldots, \tilde{\Gamma}_{\ell-1}, \Gamma_\ell\] at every occurrence of $\theta'$ and~$\phi'$.
    Note that if $\theta = \theta'$, then subpath $\Gamma_1$ is empty, and likewise, if $\phi = \phi'$, then $\Gamma_\ell$ is empty.
    This decomposition of the path $\Gamma$  corresponds to a decomposition of the string $Q[l, r-1]$ into substrings: 
    \[Q_1,\tilde{Q}_1, Q_2, \tilde{Q}_2, \ldots, \tilde{Q}_{\ell-1}, Q_\ell\]
    such that $Q[l,r-1] = Q_1\odot\tilde{Q}_1\odot \cdots \odot \tilde{Q}_{\ell-1}\odot Q_\ell$.
    Note, that $Q_i$ (resp.\ $\tilde{Q}_i$) is the empty string, if either the subpath $\Gamma_i$ (resp.\ $\tilde{\Gamma}_i$) is empty (can happen for $\Gamma_1$ and $\Gamma_\ell$) or if it only contains transitions labeled with $\epsilon$. Let $Q_i = Q[l_i,r_i]$ and $\tilde{Q}_i = Q[\tilde{l}_i,\tilde{r}_i]$. 
    Since $\Gamma_1$ is a path in $A_C$ from $\theta$ to $\theta'$ then $(g_{l_1},g_{r_1})\in G_{\theta,\theta'}$. Similarly, since $\Gamma_\ell$ is a path in $A_C$ from $\phi'$ to $\phi$ then $(g_{l_\ell},g_{r_{\ell}})\in G_{\phi',\phi}$. For $1 \leq i < \ell$,  $\tilde{\Gamma}_i$ is a path from $\theta'$ to $\phi'$ in $\hat{A}_C$. Note, that the only edges going out of $\theta'$ in $\hat{A}_C$ are the added paths matching $R(p)$, and the extended transition which is never used.
    Thus, $\tilde{\Gamma}_i$ is a path in $\hat{A}_C$ matching $R(p)$, so $(g_{\tilde{l}_i}, g_{\tilde{r}_i})\in G(p)$. For $1 < i < \ell$, $\Gamma_i$ is a path from $\phi'$ to $\theta'$, in which case, $(g_{l_i},g_{r_{i}})\in G_{\phi',\theta'}$. 
    It follows that, $(g_l, g_r) = (g_{l_1},g_{r_{\ell}})\in G_{\theta,\theta'} \odot G(p) \odot (G_{\phi',\theta'}\odot G(p))^* \odot  G_{\phi',\phi} $.

    It follows that  $(g_l, g_r) \in G_{\theta,\phi}\mid (G_{\theta,\theta'} \odot G(p) \odot (G_{\phi',\theta'}\odot G(p))^* \odot  G_{\phi',\phi}) = G$.
    \end{proof}

\subsection{Analysis}
We compute the parse tree and clustering in $O(m)$ time. Consider the time spent at each cluster $C$ with $m_C$ nodes. For case 1 ($C$ is a leaf cluster) we use $O(n^2m_C)$ by Lemma~\ref{lem:tnfamatchgraph}. For case 2 ($C$ is an internal cluster),   by Lemma~\ref{lem:matchgraphoperations} we use $O(n^\omega)$ time to compute $G(p)$ in step 1, and  by Lemma~\ref{lem:tnfamatchgraph} we use $O(n^2m_C)$ time to compute the match graphs $G_{\theta,\phi}$, $G_{\theta,\theta'}$, $G_{\phi',\theta'}$, and $G_{\phi',\phi}$. Finally, we use $O(n^\omega)$ time to compute $G(v)$ by Lemma~\ref{lem:matchgraphoperations}. Thus, we use $O(n^\omega + n^2m_C)$ at cluster $C$. In total, the algorithm uses $O(\sum_{C\in CS}(n^\omega + n^2m_C))=O(n^\omega k + n^2m)$ time by Lemma ~\ref{lem:clustering}. We store a constant number of match graphs for each of the $O(k)$ clusters. In total, we use $O(n^2 k/w + n + m)$ space. In summary, we have the following result. 

\begin{lemma}\label{lem:simplemain}
    Given an extended regular expression $R$ of length $m$ containing $k$ extended operators and a string $Q$ of length $n$, we can solve the extended regular expression matching problem for $R$ and $Q$ in space $O(\frac{n^2 k}{w} + n + m)$ and time $O(n^\omega k + n^2m)$.
\end{lemma}

\section{Reducing Space}\label{sec:reducingspace}
We now show how to reduce the $O(n^2k/w + n + m)$ space of the algorithm from the previous section to $O(n^2 \log k/w + n + m)$. We modify the bottom-up traversal to use a top-down recursive traversal, where we always process the child cluster with the largest subtree in the macro tree and carefully discard match graphs as the algorithm proceeds.

\subsection{Heavy Path Decomposition}
Let $M$ be the macro tree for $R$. Given a cluster $C\in CS$, define the \emph{weight} of $C$ to be the number of descendant clusters of $C$ in $M$. As in the \emph{heavy-path decomposition} of Sleator and Tarjan~\cite{ST1983}, we classify each cluster in $M$ as either \emph{heavy} or \emph{light}. The root of $M$ is light, and for each internal cluster $C$ with two child clusters, we pick a child cluster of maximum weight and classify it as heavy. The other child cluster is light. An edge to a light child is a \emph{light edge} and an edge to a heavy child is a \emph{heavy edge}. Removing the light edges, we partition $M$ into \emph{heavy paths}. 

\begin{lemma}[Sleater and Tarjan~\cite{ST1983}]\label{lem:heavypathdecomposition}
    Let $M$ be a macro tree with $\ell$ nodes. Any root-to-leaf path intersects at most $\log \ell + O(1)$ heavy paths.  
\end{lemma}

\subsection{Macro Tree Traversal}
We process the cluster as in the algorithm of Section~\ref{sec:improvedalgorithmsub}, but we modify the ordering as follows. There are three cases: 
\begin{description}
    \item[Case 1. $C$ is a leaf cluster.] We process $C$ as before.
    \item[Case 2. $C$ has a single child cluster.] We recursively process the child cluster before processing~$C$. 
    \item[Case 3. $C$ has two child clusters.] We recursively process the heavy child cluster of $C$, then recursively process the light child cluster, then process $C$.
\end{description}
After processing $C$, we discard the match graphs of any children of $C$. 

Since processing a cluster requires only the match graphs of its children, the algorithm correctly solves the problem as before. We show that this reduces the number of stored match graphs to $O(\log k)$. Consider the algorithm at the time of processing a cluster $C$, and let $L$ be the path of clusters from the root of the macro tree to $C$. By the ordering, the algorithm stores, at this point, the match graphs of the children of $C$ and, for each light cluster on $L$, the match graph of its heavy sibling. Recall from Lemma~\ref{lem:clustering} that $M$ contains $\ell = O(k)$ nodes. Hence, by Lemma~\ref{lem:heavypathdecomposition}, we thus store at most $\log \ell + O(1) = O(\log k)$ match graphs. In total, the algorithm uses $O(n^2 \log k/w + n + m)=O(n^2+m)$ space. 

\begin{theorem}\label{thm:simplesmallspace}
    Given an extended regular expression $R$ of length $m$ containing $k$ extended operators and a string $Q$ of length $n$, we can solve the extended regular expression matching problem for $R$ and $Q$ in space $O(n^2 + m)$ and time $O(n^\omega k + n^2m)$.
\end{theorem}

\section{Black-Box TNFA Simulation}\label{sec:blackbox}
We now show how to plug in any general fast algorithm that can efficiently simulate a TNFA. Recall, that given a TNFA $A$ of size $m$, a string $Q$ of length $n$, and two states $s, t$ in $A$, a \emph{TNFA simulation algorithm} outputs every prefix $i$ of $Q$ such that $Q[1,i]$ matches a path from $s$ to $t$. Let $T(n,m)$ and $S(n,m)$ denote the time and space, respectively, of the TNFA simulation algorithm. Observe that Lemma~\ref{lem:tnfamatchgraph} uses a TNFA simulation algorithm, namely Thompson's algorithm~\cite{Thompson1968}, with $T(n,m) = O(nm)$ and $S(n,m) = O(m)$, to construct the match graph in $O(n^2m)$ time. This result is used within our algorithm for the construction of match graphs from TNFAs in case 1 and in case 2, step 2. If we plug any general TNFA simulation algorithm into our algorithm, we obtain Theorem~\ref{thm:main_reduction}.


\section{Interval Operator}\label{sec:intervaloperator}
We now show how to extend our solution to efficiently support the interval operator. We do this by adding the operator to our match graphs and plugging it into our algorithm. Our solution is surprisingly simple and demonstrates the power of our framework.

Let $R$ be a regular expression. For integers $a$ and $b$, $0 \leq a \leq b$, the \emph{interval} operator, $R^{[a,b]}$ specifies the set of strings that consists of $\ell$ concatenations of $R$, where $\ell \in [a, b]$. Thus, the language $L(R^{[a,b]})$ is given by $L(R^{[a,b]}) = \bigcup_{a \leq i \leq b} L(R)^i$ (recall that $L(R)^i$ is the set of strings from $R$ concatenated with itself $i$ times). As a shorthand, we write $R^a$ to denote $R^{[a,a]}$. We have the following properties of the interval operator that follow immediately from the definition. 
\begin{lemma}\label{lem:intervalproperties}
    Let $R$ be an extended regular expression and let $a$ and $b$ be integers such that $0 \leq a \leq b$. Then, 
    \begin{itemize}
        \item[(i)] $R^{[a,b]} = R^{a} \odot R^{[0,b-a]}$, 
        \item[(ii)] $R^a \odot R^b = R^{a+b}$, and 
        \item[(iii)] $R^{[0,a]} \odot R^{[0,b]} = R^{[0, a + b]}$.
    \end{itemize}  
\end{lemma}
We extend our match graphs from Section~\ref{sec:matchgraphsgeneral} to include the interval operator. Let $G$ be a match graph on vertices $g_1, \ldots, g_{n+1}$. For the match graph $G^{[a,b]}$, we add an edge $(g_i, g_j)$ for each pair $i,j$ with a directed path from $g_i$ to $g_j$ in $G$ with length $\ell$, where $a \leq \ell \leq b$. Using the properties of Lemma~\ref{lem:intervalproperties} and repeated doubling, we show how to compute $G^{[a,b]}$ from $G$ with $O(\log b)$ matrix multiplications. 

\begin{lemma}\label{lem:matchgraphinterval}
Let $G$ be a match graph on a string of length $n$ and let $a$ and $b$ be integers such that $0 \leq a \leq b$. Then, we can store $G$ in $O(n^2/w + 1)$ space and compute $G^{[a,b]}$ in $O(n^\omega \log b)$ time, for any integers $a$ and $b$, $0 \leq a \leq b$. 
\end{lemma}

\begin{proof}
    Let $G$, $a$, and $b$ be given as above. As in Section~\ref{sec:dynamicprogramming}, we store match graphs as boolean matrices in $O(n^2/w + 1)$ space. Let $j$ be the largest integer such that $2^j \leq a$. We first compute $G^{a}$. To do so, we compute the sequence of match graphs $G^{1}$, $G^{2}$,$G^{4}$ \ldots, $G^{2^j}$ from $G$ as follows. We have $G^1 = G$ and for $i = 1, \ldots j$, we compute 
    \begin{equation}\label{eq:intervalpart1}
    G^{2^i} = G^{2^{i-1}} \odot G^{2^{i-1}}.
\end{equation}
    Now, let $r_1, \ldots, r_z$ be the indices of the $1$ bits in the binary representation of $a$, and compute  
\begin{equation}\label{eq:intervalpart2}
    G^{a} = \bigodot_{1 \leq i \leq z} G^{2^{r_i}}.
\end{equation}
Correctness for \eqref{eq:intervalpart1} and \eqref{eq:intervalpart2} follows from Lemma~\ref{lem:intervalproperties}(ii). 

We compute $G^{[0, b-a]}$ using the same approach. Let $j'$ be the largest integer such that $2^{j'} \leq b - a$. We compute the sequence of match graphs $G^{[0,1]}$, $G^{[0,2]}$, \ldots, $G^{[0,2^{j'}]}$, where we construct $G^{[0,1]}$ from $G$ by adding self-loops to each node, and $G^{[0, 2^{i}]}$, for $i = 1, \ldots, j'$, as $G^{[0, 2^{i}]} = G^{[0, 2^{i-1}]} \odot G^{[0, 2^{i-1}]}$ using Lemma~\ref{lem:intervalproperties}(iii). We then compute $G^{[0, b-a]}$ by concatenating the match graphs of the match graph sequence corresponding to the binary representation of $b-a$ as in \eqref{eq:intervalpart2}. Finally, we compute $G^{[a,b]} = G^{a} \odot G^{[0,b-a]}$ according to Lemma~\ref{lem:intervalproperties}(i). 

To compute $G^a$ and $G^{[0,b -a]}$ we use $O(\log a + \log (b-a)) = O(\log b)$ match graph concatenations. The remaining step of the algorithm uses a constant number of operations. By Lemma~\ref{lem:matchgraphoperations}, the total time is $O(n^\omega \log b)$. We can avoid storing the $O(j)$ match graphs during the computation of $G^a$  by computing \eqref{eq:intervalpart2} in parallel with $G^{1}$, $G^{2}$, \ldots, $G^{2^j}$ and concatenating $G^{2^{r_i}}$ in \eqref{eq:intervalpart2} as soon as it is computed.

This way, we only store a constant number of previously computed match graphs. The same applies when computing $G^{[0,b -a]}$. Thus, the total space is $O(n^2/w + 1)$. 
\end{proof}

We plug the match graph construction from Lemma~\ref{lem:matchgraphinterval} for the interval operator into our main algorithm of Section~\ref{sec:improvedalgorithmsub} whenever we compute $G(p)$ in step 1. All of the other steps remain the same. The further techniques in Sections~\ref{sec:reducingspace} and~\ref{sec:blackbox} still apply. In summary, we have shown the following result. 

\begin{theorem}\label{thm:maininterval}
Given an extended regular expression $R$ of length $m$ containing $k$ extended operators, $k'$ interval operators, and a string $Q$ of length $n$, we can solve the extended regular expression matching problem for $R$ and $Q$ in space $O(\frac{n^2 \log (k+k')}{w} + n + m) = O(n^2 +m)$ and time 
\[
O\left(n^\omega (k + k'\log B) + \frac{n^2m}{\max(w/\log w, \log (n+m))} + m\right).
\]
Here, $B$ is the maximum of the upper bounds in ranges for the $k'$ interval operators in $R$. 
\end{theorem}
Thus, compared to Theorem~\ref{thm:main}, we only incur an additional $\log B$ factor on each of the $k'$ interval operators.

\section{Concluding remarks}
In this work, we have given an algorithm for extended regular expression matching, which uses simple techniques to improve the classical dynamic programming solution.
We also show how the framework can be used as a black-box to plug in any TNFA simulation algorithm. If we plug in the state-of-the-art, then our algorithm outperforms existing algorithms both in terms of time and space.
Lastly, we show that our framework is easily extendable to other operators as one simply has to define the corresponding match graph operation. It would be interesting to consider extending the regular expressions further, by adding more operators to the framework.

\bibliographystyle{plain}
\bibliography{references}

\end{document}